\documentclass[runningheads,11pt]{llncs}
\usepackage{cite}
\usepackage{amssymb,amsmath,amsfonts}
  \usepackage{geometry}
  \newgeometry{vmargin={35mm}, hmargin={35mm}}
\usepackage{graphicx}
\usepackage{subfigmat}
\usepackage{color}
\usepackage{url}
\usepackage{latexsym}
\usepackage{wasysym}
\usepackage{newlfont}
\usepackage{enumerate}
\usepackage{stmaryrd}
\numberwithin{equation}{section}
    \urldef{\mailsa}\path|stephen.moore@ucc.edu.gh|
    \urldef{\mailsb}\path|eric.okyere@uenr.edu.gh|
    \newcommand{\keywords}[1]{\par\addvspace\baselineskip
     \noindent\keywordname\enspace\ignorespaces#1}
\begin{document}
\mainmatter
\title{Controlling the Transmission Dynamics of COVID-19}
\titlerunning{Controlling the Transmission Dynamics of COVID-19}
\authorrunning{S.~E.~Moore and E.~Okyere}
\author{Stephen E.~Moore\inst{1}\thanks{Corresponding author : stephen.moore@ucc.edu.gh} \and Eric Okyere\inst{2}}
\institute{Department of Mathematics, \\
University of Cape Coast, Ghana, \\
\and
Department of Mathematics and Statistics, \\
University of Energy and Natural Resources, Ghana.\
}
% }
\maketitle
%-------------------------------------------------------------------------------------
\begin{abstract}
The outbreak of COVID-19 caused by SARS-CoV-2 in Wuhan and other cities in China in 2019 has become a global pandemic
as declared by World Health Organization (WHO) in the first quarter of 2020 \cite{Who:2020a}.
The delay in diagnosis, limited hospital resources
and other treatment resources leads to rapid spread of COVID-19. In this article, we
consider an optimal control COVID-19 transmission model and assess the impact of some
control measures that can lead to the reduction of exposed and infectious individuals in
the population. We investigate three control strategies for this deadly infectious disease
using personal protection, treatment with early diagnosis, treatment with delay diagnosis
and spraying of virus in the environment as time-dependent control functions in our
dynamical model to curb the disease spread.

\end{abstract}
%%%%%%%%%%%%%%%%%%%%%%%%%%%%%%%%%%%%%%%%%%%%%%%%%%%%%%%%%%%%%%%%%%%%%%%%%%%%%%%%%%%%%%
\,\,
%/--------------------------------------------------------------------------------------------------------
\section{Introduction}
\label{sec:Introduction}
The recent outbreak of the deadly and highly infectious COVID-19 disease caused by SARS-CoV-2 in Wuhan and other cities in China in 2019 has become a global pandemic as declared by World Health Organization (WHO) in the first quarter of 2020 \cite{Who:2020a}. The most vulnerable people to develop serious complications from this dangerous disease are the elderly with underlying medical problems. As at 29th March, 2020, the 69th situation report by WHO indicated that the deadly COVID-19 disease had globally infected $634,835$ people with $29,891$ deaths, see \cite{Who:2020b}. 

The understanding of the transmission dynamics of infectious diseases has been well-studied and researched in mathematics and usually referred to
as mathematical epidemiology. These mathematical models have played a major role in increasing understanding of the underlying mechanisms which
influence the spread of diseases and provide guidelines as to how the spread can be controlled \cite{brauer2017mathematical, siettos2013mathematical, hethcote2000mathematics}. The recent outbreak of the deadly and highly infectious COVID-19 disease has attracted the attention of many authors who have discussed and studied the nature of the virus, its transmission dynamics and the basic reproduction number of the disease, see eg. \cite{chen2020mathematical, zhang2020estimation, FangNiePenny:2020, tang2020updated, wang2020phase, jia2020extended}. Recently, Elsevier and Springer have made open access to several literature for interested researchers \cite{Elsevier:2020, Springer:2020}.

An SEIR mathematical model for the transmission dynamics of COVID-19 disease with data fitting, parameter estimations and sensitivity analysis is studied in \cite{FangNiePenny:2020} whiles a deterministic model for COVID-19 that captures the effect of delay diagnosis on the disease transmission has also been presented, see \cite{XinmiaoLiuHuidiMeng:2020}. In \cite{MizumotoChowell:2020}, the authors explore a statistical analysis of COVID-19 disease data to estimate time-delay adjusted risk for death from this deadly virus in Wuhan, as well as for China excluding Wuhan. Their study suggested that effective social distancing and movement restrictions practices can help minimise the disease transmission. A real-time forecast phenomenological model has also been designed to study the transmission pattern of COVID-19 infectious disease, see e.g., \cite{roosa2020real}. Also, an SEIR-type compartmental modelling concept applied to design a data-driven epidemic model that incorporates governmental actions and individual behavioural reactions for the COVID-19 disease outbreak in Wuhan \cite{lin2020conceptual}.

Mathematical modelling of epidemics using deterministic optimal control problems is widely explored in the literature of mathematical epidemiology. A detailed comprehensive literature of optimal control models in epidemiological modeling and numerical approximation techniques can be found in \cite{lenhart2007optimal,sharomi2017optimal}. Several works in literature reveals that epidemic models that are constructed with optimal control problems are appropriate and very useful for suggesting control strategies to curb disease spread, see, e.g., \cite{momoh2018optimal, okyere2019deterministic, bonyah2019modelling, olaniyi2018global}.

The principal purpose of this article is to present control strategies for transmission dynamics of
COVID-19 and to determine strategies that are critical even during instances of delay in diagnosis. In Section~\ref{sec:modelproblemandcontrolmodelproblem}, we formulate an optimal control model for COVID-19 with four control measures. The analysis of the control model is presented in Section~\ref{sec:mathematicalanalysisofthemodel}. In Section~\ref{sec:numericalresults}, we present the numerical results of the optimal control model. Finally, we conlude in Section~\ref{sec:conclusion} with discussions on the control measures.
{\allowdisplaybreaks
%----------------------------------------------------
\section{Formulation of the Optimal Control Problem}
\label{sec:modelproblemandcontrolmodelproblem}
In this section, we formulate an optimal control model for COVID-19 to derive four control measures with minimal
implementation cost to eradicate the disease after a defined period of time. Our new epidemiological time-dependent 
control model is an extended and modified version of the COVID-19 transmission dynamical model introduced in \cite{XinmiaoLiuHuidiMeng:2020}.
Here, we note that the population is divided into susceptible $(S),$ self-quarantine susceptible $(S_q),$ exposed $(E),$
infectious with timely diagnosis $(I_1),$ infectious with delayed diagnosis $(I_2),$ hospitalized $(H),$ recovered $(R)$
and the viral spread in the environment $(V)$. Following the compartmental transition diagram as shown in Figure~\ref{fig:transmissiondynamicsmodel}, the 
eight-state dynamical model describing the transmission dynamics of COVID-19 is given by
\begin{figure}[htb!]
\centering
 \includegraphics[width=0.8\textwidth]{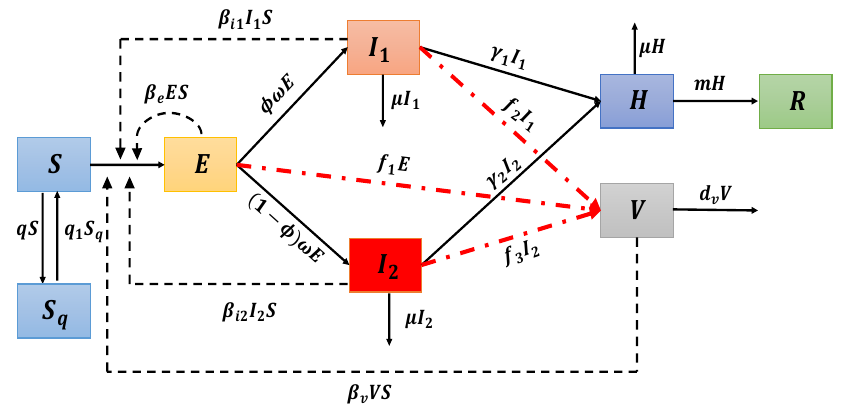}
\caption{Compartmental diagram for the transmission dynamics of COVID-19, see \cite{XinmiaoLiuHuidiMeng:2020}.}
\label{fig:transmissiondynamicsmodel}
\end{figure}
%
%-------------
\begin{align}
\label{eqn:modelproblem}
\frac{dS}{dt}&=-(\beta_{e}E+\beta_{i_1}I_1+\beta_{i_2}I_2+\beta_{v}V)S-qS+q_{1}S_q \notag \\
\frac{d S_{q}}{dt}&=qS-q_{1}S_{q} \notag  \\
\frac{dE}{dt}&=(\beta_{e}E+\beta_{i_1}I_1+\beta_{i_2}I_2+\beta_{v}V)S- \omega E \notag  \\
\frac{dI_{1}}{dt}&=\phi \omega E-\gamma_{1}I_1-\mu I_{1}  \\
\frac{dI_{2}}{dt}&=(1-\phi) \omega E-\gamma_{2}I_2-\mu I_{2}\notag  \\
\frac{d H}{dt}&=\gamma_{1}I_1+\gamma_{2}I_2- m H-\mu H\notag  \\
\frac{dR}{dt}&=m H\notag  \\
\frac{dV}{dt}&=f_1 E+f_2 I_1+f_3 I_2-d_v V\notag,
\end{align}
where $\beta_e,\beta_{i_1}, \beta_{i_2}$ and $\beta_v$ denote the transmission rates from the exposed, infectious with timely diagnosis,
infectious with delay diagnosis and virus in the environment to the susceptible, respectively. 
Also $f_i, i=1,2,3$ is the rate of virus in the environment from both the exposed and the infectious and removed at rate $d_v.$ 
For the rest of the model parameters, we refer the interested reader to reference \cite{XinmiaoLiuHuidiMeng:2020} for detailed descriptions.
\subsection{COVID-19 Model Problem with Control Measures}
\label{subsec:covid19modelwithcontrols}
Following from the system \eqref{eqn:modelproblem}, we modified the transmission rate by reducing the factor by $(1-u_1)$,
where $u_1$ measures the effort of individuals to protect themselves (i.e. personal protection). The control variable $u_2$
measures the treatment rate of timely diagnosed individuals whiles the $u_3$ measures the treatment rate of delayed diagnosed individuals.
We assume that $u_2I_1$  and $u_3I_2$ individuals are removed from the timely diagnosed class and delayed diagnosed class and added to the
Hospitalized class. The fourth control variable $u_4$ measures the spraying of the environment to prevent viral release.
We also assume that $u_4V$ virus are removed from the environment. We further assume that individuals that recovers at any 
time $t$ after hospitalization and treatment are removed from the hospitalized class to the recovered class. 
With regards to these assumptions, the dynamics of system \eqref{eqn:modelproblem} are modified into the following system of equations:
%------------
\begin{align}
\label{eqn:covid19modelwithcontrols}
\frac{dS}{dt}&=-(1-u_1)(\beta_{e}E+\beta_{i_1}I_1+\beta_{i_2}I_2+\beta_{v}V)S-qS+q_{1}S_q\notag  \\
\frac{d S_{q}}{dt}&=qS-q_{1}S_{q} \notag \\
\frac{dE}{dt}&=(1-u_1)(\beta_{e}E+\beta_{i_1}I_1+\beta_{i_2}I_2+\beta_{v}V)S- \omega E \notag  \\
\frac{dI_{1}}{dt}&=\phi \omega E-u_{2}I_1-\mu I_{1}  \\
\frac{dI_{2}}{dt}&=(1-\phi)\omega E-u_{3}I_2-\mu I_{2}\notag  \\
\frac{d{H}}{dt}&=u_{2}I_1+u_{3}I_2- m H-\mu H \notag  \\
\frac{dR}{dt}&=m H \notag  \\
\frac{dV}{dt}&=f_1 E+f_2 I_1+f_3 I_2-d_v V-u_{4}V, \notag
\end{align}
where $q$ is the rate at which susceptible $(S)$ individuals move into self-quarantine $(S_q)$ and 
$q_1$ is the rate at which self-quarantined individuals become susceptible again. Also,  $\mu$ is
the death rate and hospitalized individuals are decreased at recovery rate $m.$ 
}
%%%%%%%%%%%%%%%%
\section{Mathematical Analysis of the Model}
\label{sec:mathematicalanalysisofthemodel}
 In this section, we will formulate an objective functional and present the existence of
 optimal control by means of Pontryagin's Maximum Principle.
 Given the optimal control problem \eqref{eqn:covid19modelwithcontrols}, we prove the existence of
 control problem following \cite{siettos2013mathematical} and then characterizing it for
optimality.
%%%%%%%%%%%%%%%%%%%%%%%%%%
% \textcolor{red}{
The objective functional $\mathcal{J}$ formulates the optimization problem of identifying the most effective strategies.
The overall preselected objective involves the minimization of the number of exposed, delayed diagnosed infectious individuals and the viral spread in the environment over a finite time interval $[0,T].$
% }
We define the objective functional $\mathcal{J},$ as follows
\begin{equation}
 \label{eqn:objectivefunctional}
\mathcal{J}(u_{1},u_{2},u_{3},u_{4}):=\int_{0}^{T}\bigg(A_{1}E+A_{2}I_2+A_3V+\dfrac{1}{2}\sum_{i=1}^{4}C_{i}u^{2}_{i}(t)\bigg)dt.
\end{equation}
We aim to minimize the cost functional \eqref{eqn:objectivefunctional} which includes the number of exposed $(E),$
infectious with delay diagnosis $(I_2),$ and the virus in the environment $(V),$
as well as the social costs related to the resources needed for personal protection $C_1u^2_1,$
early detected treatment $C_2u^2_2,$ delay detected treatment $C_3u^2_3,$ and spraying of
environment $C_4u^2_4.$
The control effort is modeled by means of a linear combination of quadratic terms, $u_i^2(t), i=1,\ldots,4.$
The constants $A_j, j=1,\ldots,3$ and $C_i, i=1,\ldots,4$ represent a measure of the relative cost of the
interventions over time $[0,T].$ 

The objective of the control problem is to seek 
% control 
functions $(u^*_1(t),u^*_2(t),u^*_3(t),u^*_4(t))$ such that
\begin{equation}
\label{eqn:optimalcontrolsolution}
 \mathcal{J}(u^*_1(t),u^*_2(t),u^*_3(t),u^*_4(t)) := \min \{ \mathcal{J}(u_1,u_2,u_3,u_4), (u_1,u_2,u_3,u_4) \in \mathcal{U} \},
\end{equation}
where the control set is defined as
\begin{align}
 \label{eqn:controlset}
 \mathcal{U} := \{u=(u_1,u_2,u_3,u_4) | u_i(t) \, &\text{is Lebesgue measurable,} \, \notag \\
             & 0 \leq u_i(t)\leq 1, t \in[0,T] \quad \text{for} \quad i=1,\ldots,4\},
\end{align}
subject to the COVID-19 model with controls \eqref{eqn:covid19modelwithcontrols} and appropriate initial conditions.
%%%%%%%%%%%%%%%%%%%%%%%%%%%%%%%%%%%%%%%%%%%%%%%%%%

In the next section, we prove the existence of an optimal control for the system \eqref{eqn:covid19modelwithcontrols}
and then derive the optimality system. It is well known that Pontryagin's maximum principle (PMP) is required to solve
this control problem and the derivation of the necessary conditions \cite{Pontryagin:1985,Pontryagin:1962}.

\subsection{Existence of an Optimal Control}
\label{subsec:existenceoptimalcontrolmodel}
The necessary conditions include the optimality solutions and the adjoint equations that an optimal control must
satisfy which come from Pontryagin's maximum principle \cite{Pontryagin:1962}.
This principle converts the control model \eqref{eqn:covid19modelwithcontrols} and the objective functional \eqref{eqn:objectivefunctional}
into a problem of minimizing pointwise Hamiltonian function \eqref{eqn:hamiltonian}, which is formed by allowing each of the adjoint
variables to correspond to each of the state variables accordingly and combining the results with the objective functional.

\begin{theorem}
 \label{thm:existenceoptimalcontrol}
 Given the objective functional $\mathcal{J}(u_1,u_2,u_3,u_4)$ as in \eqref{eqn:objectivefunctional}, where the control
 set $\mathcal{U}$ given by \eqref{eqn:controlset} is measurable subject to \eqref{eqn:covid19modelwithcontrols} with
 initial conditions given at $t=0,$ then there exists an optimal control $u^*=(u^*_1(t),u^*_2(t),u^*_3(t),u^*_4(t))$ such
 that $\mathcal{J}(u^*_1(t),u^*_2(t),u^*_3(t),u^*_4(t)) := \min \{ \mathcal{J}(u_1,u_2,u_3,u_4), (u_1,u_2,u_3,u_4) \in \mathcal{U}\}.$
\end{theorem}
\begin{proof}
The existence of an optimal control due to the convexity of the integrand of $\mathcal{J}$
with respect to the control measures $u_i, i=1,\ldots,4,$ an \textit{a priori} boundedness of the solutions
of both the state and adjoint equations and the Lipchitz property of the state system with
respect to the state variables follows from \cite{fleming2012deterministic}. \qed
\end{proof}

{\allowdisplaybreaks
To find the optimal solution, we need the Lagrangian $(\mathbf{L})$ and Hamiltonian $(\mathbf{H})$ for the optimal control
problem \eqref{eqn:covid19modelwithcontrols} and \eqref{eqn:objectivefunctional}. The Lagrangian of the control problem is given by
\begin{align}
 \label{eqn:lagrnagian}
 \mathbf{L} := A_{1}E+A_{2}I_2++A_3V+\dfrac{1}{2}\sum_{i=1}^{4}C_{i}u^{2}_{i}(t).
\end{align}
Since we want the minimal value of the Lagrangian, we define the Hamiltonian function for the system as
\begin{align}
\label{eqn:hamiltonian}
\mathbf{H}&=A_{1}E+A_{2}I_2+A_3V+\dfrac{1}{2}\big(C_{1}u_{1}^{2}+C_{2}u_{2}^{2}+C_{3}u_{3}^{2}+C_{4}u_{4}^{2}\big)\nonumber\\
&+\lambda_{S}\big[-(1-u_1)(\beta_{e}E+\beta_{i_1}I_1+\beta_{i_2}I_2+\beta_{v}V)S-qS+q_{1}S_q\big]\nonumber\\
&+\lambda_{S_q}\big[qS-q_{1}S_{q}\big]%\nonumber\\
+\lambda_{E}\big[(1-u_1)(\beta_{e}E+\beta_{i_1}I_1+\beta_{i_2}I_2+\beta_{v}V)S-wE\big] \\
&+\lambda_{I_1}\big[\phi wE-u_{2}I_1-\mu I_{1}]
+\lambda_{I_2}\big[(1-\phi)wE-u_{3}I_2-\mu I_{2}\big]\nonumber\\
&+\lambda_{H}\big[u_{2}I_1+u_{3}I_2- m H-\mu H\big]\nonumber\\
&+\lambda_{R}m H%\nonumber\\
+\lambda_{V}\big[f_1 E+f_2 I_1+f_3 I_2-d_v V-u_{4}V\big]\nonumber,
\end{align}
where $\lambda_j, j\in \{S,S_q,E,I_1,I_2,H,R,V\}$ are the adjoint variables.
}
Next, we apply the necessary conditions to the Hamiltonian $\mathbf{H}$ in \eqref{eqn:hamiltonian}.
%%%%%%%%%%%%%%%%5
% \newpage
\begin{theorem}
 \label{thm:controlfunctionexistence}
  Given an optimal control $u^*:= (u^*_1,u^*_2,u^*_3,u^*_4)$ and a solution\\ $y^*=(S^*,S^*_q,E^*,I_1^*,I_2^*,H^*,R^*,V^*)$ of the
  corresponding state system \eqref{eqn:covid19modelwithcontrols}, there exists adjoint variable $\lambda_j, j\in \{S,S_q,E,I_1,I_2,H,R,V\}$
  satisfying
{\allowdisplaybreaks

\begin{align}
\label{eqn:adjointsystem}
\dfrac{d\lambda_{S}}{dt}&=(\lambda_{S}-\lambda_{E})(1-u_1)\bigg[\beta_{e}E+\beta_{i_1}I_1+\beta_{i_2}I_2+\beta_{v}V\bigg]+q(\lambda_{S}-\lambda_{E})\nonumber\\
\dfrac{d\lambda_{S_q}}{dt}&=q_{1}(\lambda_{S_q}-\lambda_{S})\nonumber\\
\dfrac{d\lambda_{E}}{dt}&=-A_1+(\lambda_{S}-\lambda_{E})(1-u_{1})\beta_{e}S+w\lambda_{E}-\phi w\lambda_{I_1}-(1-\phi)\lambda_{I_2}-\tau_{1}\lambda_{V}\nonumber\\
\dfrac{d\lambda_{I_1}}{dt}&=(\lambda_{S}-\lambda_{E})(1-u_{1})\beta_{i_1}S+(\lambda_{I_1}-\lambda_{H})u_2+\lambda_{I_1}\mu-\tau_{2}\lambda_{V}\\
\dfrac{d\lambda_{I_2}}{dt}&=-A_2+(\lambda_{S}-\lambda_{E})(1-u_{1})\beta_{i_2}S+(\lambda_{I_2}-\lambda_{H})u_3+\lambda_{I_2}\mu-\tau_{3}\lambda_{V}\nonumber\\
\dfrac{d\lambda_{H}}{dt}&=m(\lambda_{H}-\lambda_{R})+\mu\lambda_{H}\nonumber\\
\dfrac{d\lambda_{R}}{dt}&=0\nonumber\\
\dfrac{d\lambda_{V}}{dt}&=-A_3+(\lambda_{S}-\lambda_{E})(1-u_{1})\beta_{v}S+\lambda_{V}(d_v+u_4)\nonumber
\end{align}
with transversality conditions
\begin{equation}
 \label{eqn:transversalityconditions}
 \lambda_j(T) =0, \quad  j\in \{S,S_q,E,I_1,I_2,H,R,V\}.
\end{equation}

Furthermore, the control functions $u^*_1,u^*_2,u^*_3$ and $u^*_4$ are given by
\begin{align*}
u^{*}_{1}&=\min\{1, \max\{0, \Lambda_1\} \},  \\
u^{*}_{2}&= \min\{1, \max\{0, \Lambda_2\} \},\\
u^{*}_{3}&=\min\{1, \max\{0, \Lambda_3\} \}, \quad \text{and} \quad \\
u^{*}_{4}&=\min\{1, \max\{0, \Lambda_4\} \},
\end{align*}
where
\begin{align}
 \label{eqn:controlfunctionconditions}
   \Lambda_1 & = \dfrac{\big(\lambda_{E}-\lambda_{S}\big)(\beta_{e}E+\beta_{i_1}I_1+\beta_{i_2}I_2+\beta_{v}V)S}{C_{1}}, \notag \\
   \Lambda_2 & = \dfrac{\big(\lambda_{I_1}-\lambda_{H}\big)I_1}{C_{2}}, \quad
   \Lambda_3  = \dfrac{\big(\lambda_{I_2}-\lambda_{H}\big)I_2}{C_{3}} \quad \text{and} \quad
   \Lambda_4  = \dfrac{\lambda_{V}V}{C_{4}}.
\end{align}

}
\end{theorem}
\begin{proof}
The proof methodology follows similarly as presented by \cite{lenhart2007optimal}. \qed
\end{proof}

{\allowdisplaybreaks

%%%%%%%%%%%%%%%%%%%%%%%%%%%%%%%%%%%%%%%
\section{Numerical Results of the Optimal Control Model and Discussion}
\label{sec:numericalresults}
In this section, we present the numerical solutions for our optimality problem using the fourth-order runge-kutta forward-backward sweep method. This numerical scheme is  very efficient and has been widely used by several authors in simulating their optimal control problems \cite{hugo2017optimal, neilan2010modeling, okosun2014co, bonyah2019modelling, okyere2019deterministic}. The details of this scheme can be found in the monograph \cite{lenhart2007optimal}. Parameter values for our numerical illustrations are adapted from \cite{XinmiaoLiuHuidiMeng:2020} where the authors used data from the city of Wuhan in the Hubei province of China. We assume $A_1=5, A_2=5, A_3=10, C_1=10, C_2=30, C_3=25$ and $C_4=30$. Figure~\ref{fig:optimalcontrolfunction} below represent profiles of the optimal control functions $(u_1, u_2, u_3, u_4).$

\begin{figure}[!htb]
  \centering
  \includegraphics[width=0.5\textwidth]{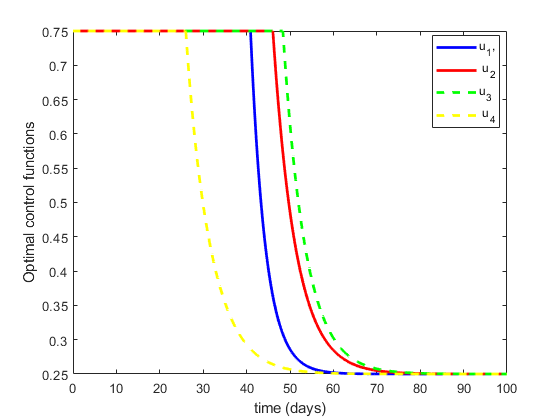}
  \caption{Optimal control functions}
  \label{fig:optimalcontrolfunction}
\end{figure}

\subsection{Control Strategy I}
In this strategy, we consider personal protection $(u_1\neq 0)$, treatment with early diagnosis $(u_2\neq 0)$, treatment with delay diagnosis $(u_3\neq 0)$ time-dependent control functions to minimise our objective functional. Our main aim in this control strategy is to minimise the number of exposed $(E)$, infectious with delay diagnosis $(I_2)$ and virus in the environment $(V)$. In the non-optimal control model~\eqref{eqn:covid19modelwithcontrols}, treatment with early diagnosis $(\gamma_1=u_2=\frac{1}{2.9})$ and treatment with delay diagnosis $(\gamma_2=u_3=\frac{1}{10})$ are captured as constant controls.

\begin{figure}[!htbp]
  \centering
  \includegraphics[width=0.3\textwidth]{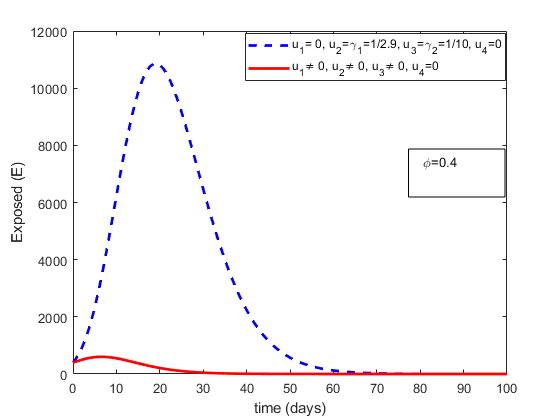}
  \includegraphics[width=0.3\textwidth]{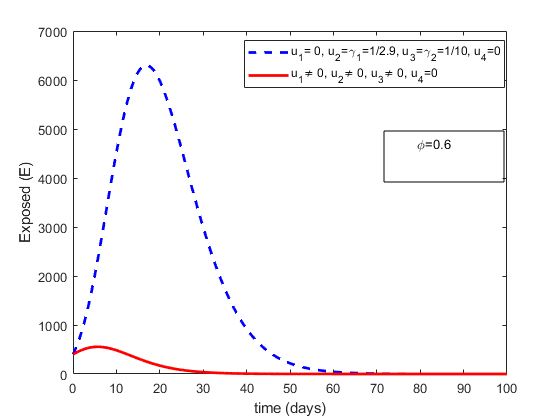}
  \includegraphics[width=0.3\textwidth]{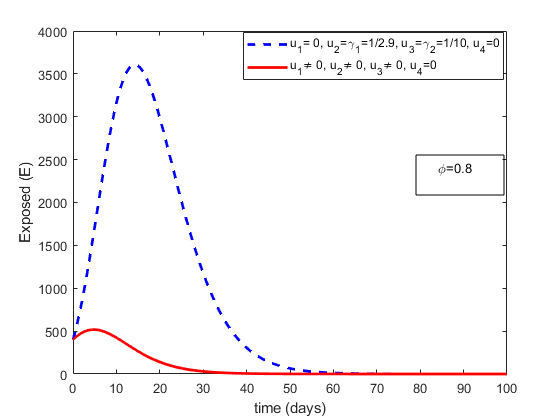}
  \caption{Solution trajectories for Exposed individuals with varying parameter
   $\phi =0.4$, $\phi =0.6$ and $\phi=0.8.$ The red line represents the controlled Exposed population
   whiles the blue line represents the uncontrolled exposed population.}
  \label{fig:exposed_A}
\end{figure}

\begin{figure}[!htbp]
  \centering
  \includegraphics[width=0.3\textwidth]{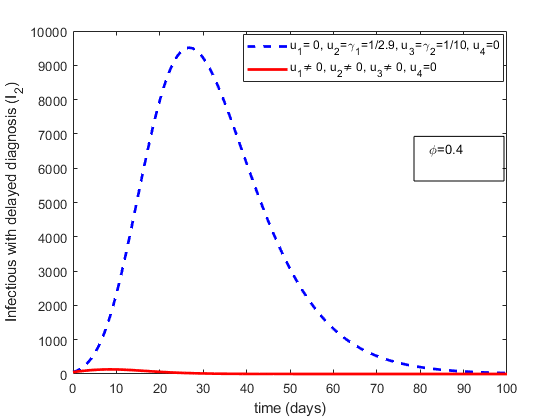}
  \includegraphics[width=0.3\textwidth]{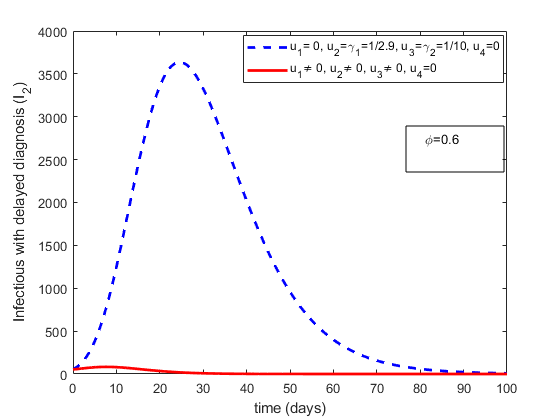}
  \includegraphics[width=0.3\textwidth]{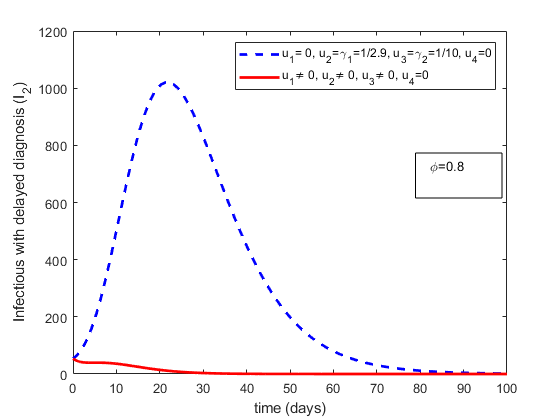}
  \caption{Solution trajectories for Infectious individuals with delayed diagnosis with varying parameter
   $\phi =0.4$, $\phi =0.6$ and $\phi=0.8.$ The red line represents the controlled delayed diagnosed infectious population
   whiles the blue line represents the uncontrolled infectious population.}
  \label{fig:infectiousdelaydiagnosis_A}
\end{figure}

\begin{figure}[!htbp]
  \centering
  \includegraphics[width=0.3\textwidth]{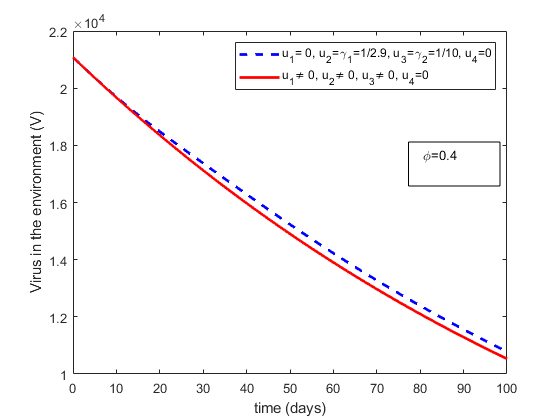}
  \includegraphics[width=0.3\textwidth]{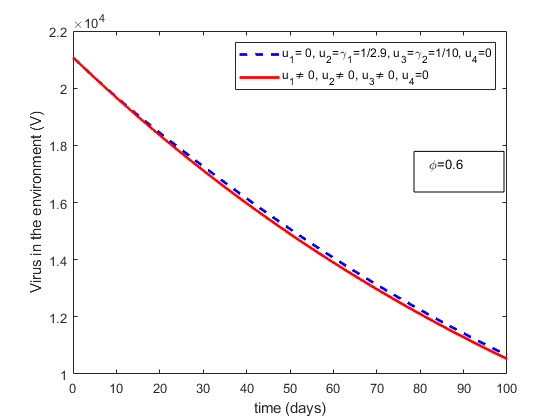}
  \includegraphics[width=0.3\textwidth]{Virus_in_Environment_04A}
  \caption{Controlling the viral spread in the environment with varying proportion of sympomatic individuals
  $\phi =0.4$, $\phi =0.6$ and $\phi=0.8$ where the red line represents the controlled envrionment and the blue
  line represents the uncontrolled environment.}
  \label{fig:virusinenvironment_A}
\end{figure}

\newpage
\subsection{Control Strategy II}
\label{subsec:strategyII}
This strategy deals with treatment control with early diagnosis $(u_2\neq 0)$ and treatment control with delay diagnosis $(u_3\neq 0)$ to minimise our objective functional. Our main aim in this control strategy is to minimise the number of exposed $(E)$, infectious with delay diagnosis $(I_2)$ and virus in the environment $(V)$. In the non-optimal control model~\eqref{eqn:covid19modelwithcontrols}, treatment with early diagnosis $(\gamma_1=u_2=\frac{1}{2.9})$ and treatment with delay diagnosis $(\gamma_2=u_3=\frac{1}{10})$ are captured as constant controls.

\begin{figure}[!htbp]
  \centering
  \includegraphics[width=0.3\textwidth]{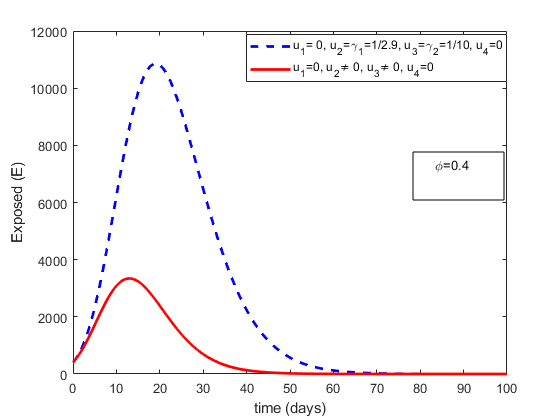}
  \includegraphics[width=0.3\textwidth]{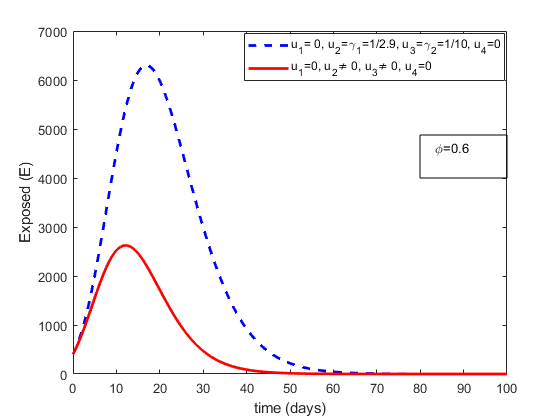}
  \includegraphics[width=0.3\textwidth]{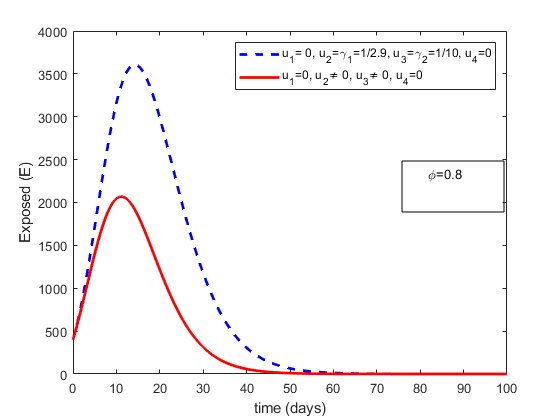}
  \caption{Solution trajectories for Exposed individuals with varying parameter
   $\phi =0.4$, $\phi =0.6$ and $\phi=0.8.$ The red line represents the controlled exposed population
   whiles the blue line represents the uncontrolled exposed population.}
  \label{fig:exposed_B}
\end{figure}

\begin{figure}[!htbp]
  \centering
  \includegraphics[width=0.3\textwidth]{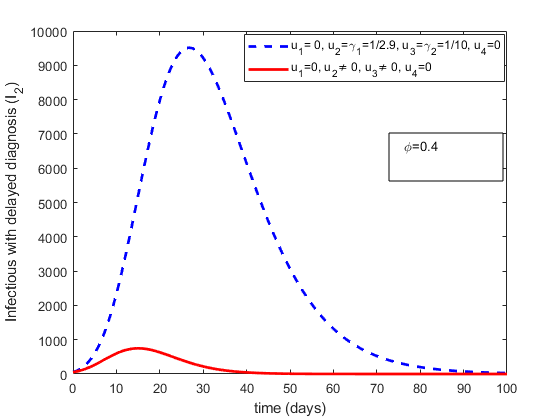}
  \includegraphics[width=0.3\textwidth]{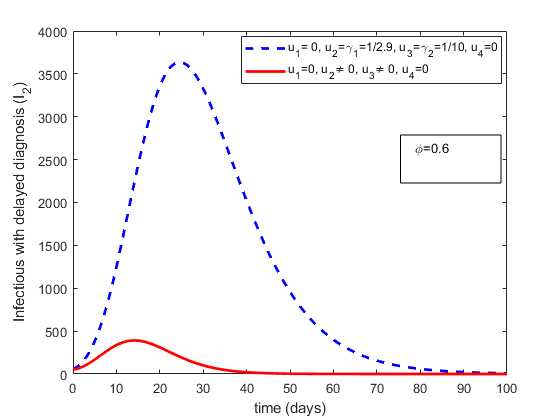}
  \includegraphics[width=0.3\textwidth]{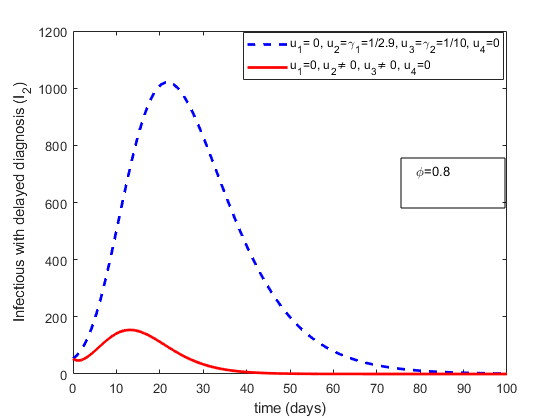}
  \caption{Solution trajectories for Infectious individuals with delayed diagnosis  with varying parameter
   $\phi =0.4$, $\phi =0.6$ and $\phi=0.8.$ The red line represents the controlled delayed diagnosed infectious population
   whiles the blue line represents the uncontrolled infectious population.}
  \label{fig:infectiousdelaydiagnosis_B}
\end{figure}

\begin{figure}[!htbp]
  \centering
  \includegraphics[width=0.3\textwidth]{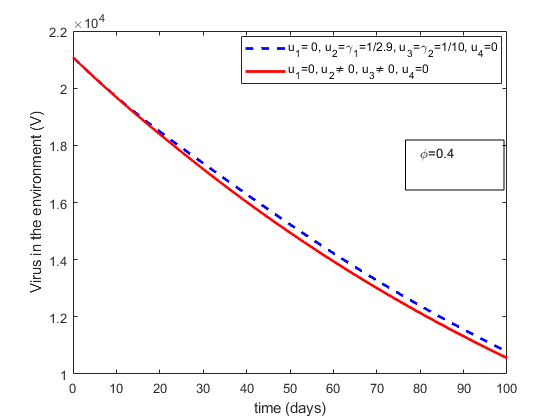}
  \includegraphics[width=0.3\textwidth]{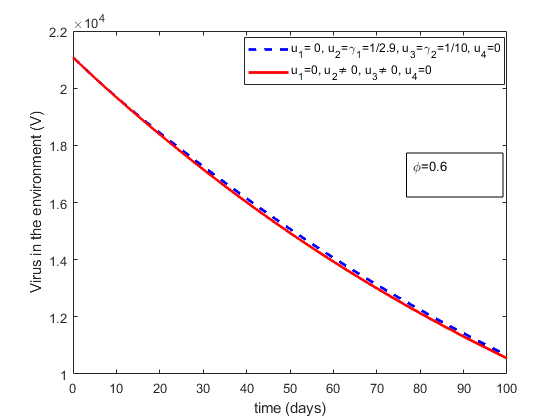}
  \includegraphics[width=0.3\textwidth]{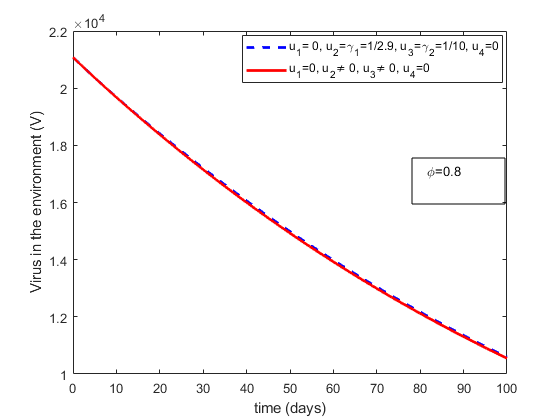}
  \caption{Controlling the viral spread in the environment with varying proportion of symptomatic individuals
  $\phi =0.4$, $\phi =0.6$ and $\phi=0.8$ where the red line represents the controlled environment and the blue
  line represents the uncontrolled environment.}
  \label{fig:virusinenvironment_B}
\end{figure}

%%%%%%%%%%%%%%%%%%%%%%%%%%%%%%%%%%%%%
\newpage
\subsection{Control Strategy III}
\label{subsec:strategyIII}
In this strategy all the four time-dependent control functions $(u_1\neq 0, u_2\neq 0, u_3\neq 0, u_4\neq 0)$ proposed in this study are incorporated into the optimal control COVID-19 model problem to minimise the objective function. In the non-optimal control model, treatment with early diagnosis $(\gamma_1=u_2=\frac{1}{2.9})$ and treatment with delay diagnosis $(\gamma_2=u_3=\frac{1}{10})$ are captured as constant controls.
\begin{figure}[!htbp]
  \centering
  \includegraphics[width=0.3\textwidth]{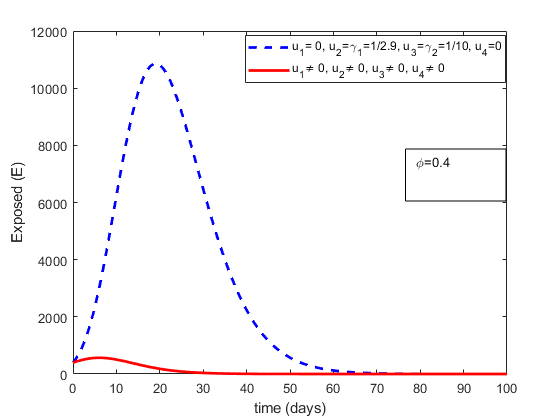}
  \includegraphics[width=0.3\textwidth]{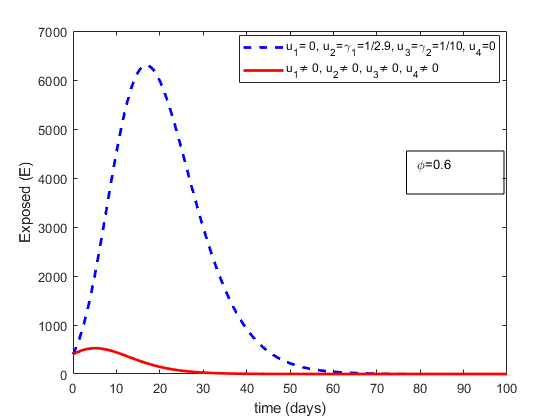}
  \includegraphics[width=0.3\textwidth]{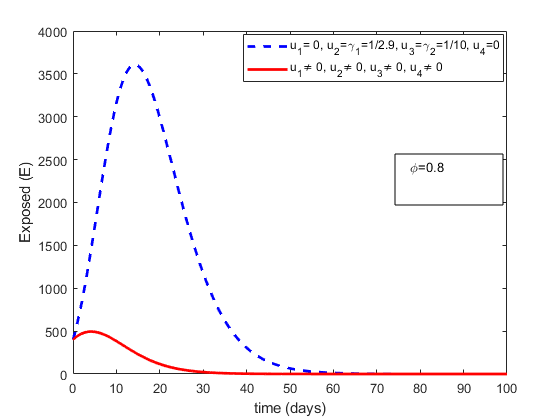}
  \caption{Solution trajectories for Exposed individuals with varying parameter
   $\phi =0.4$, $\phi =0.6$ and $\phi=0.8.$ The red line represents the controlled Exposed population
   whiles the blue line represents the uncontrolled exposed population.}
  \label{fig:exposed_C}
\end{figure}

%------------------------------------------------
\begin{figure}[!htbp]
  \centering
  \includegraphics[width=0.3\textwidth]{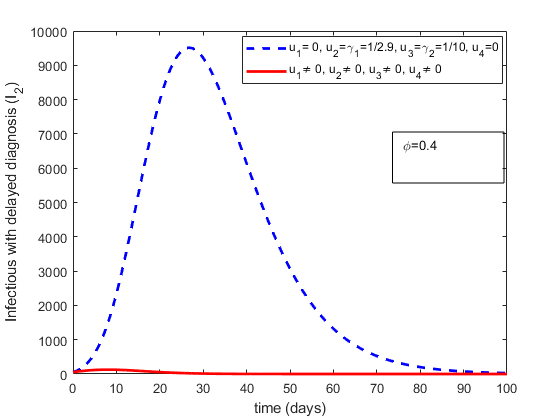}
  \includegraphics[width=0.3\textwidth]{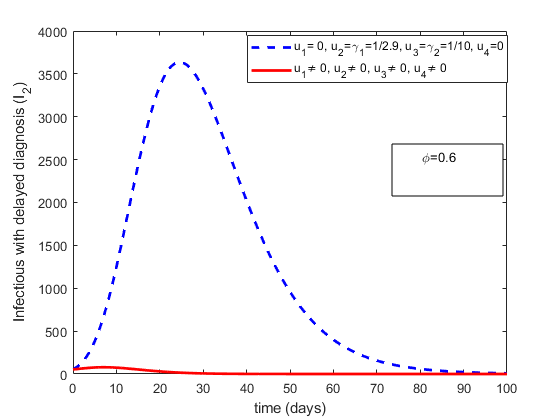}
  \includegraphics[width=0.3\textwidth]{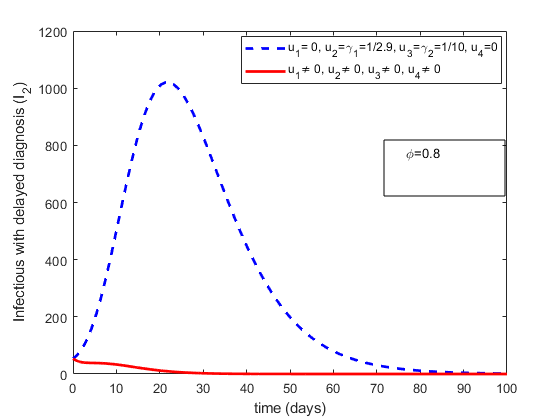}
  \caption{Solution trajectories for Infectious individuals with delayed diagnosis with varying parameter
   $\phi =0.4$, $\phi =0.6$ and $\phi=0.8.$ The red line represents the controlled delayed diagnosed infectious population
   whiles the blue line represents the uncontrolled infectious population.}
  \label{fig:infectiousdelaydiagnosis_C}
\end{figure}

\begin{figure}[!htbp]
  \centering
  \includegraphics[width=0.3\textwidth]{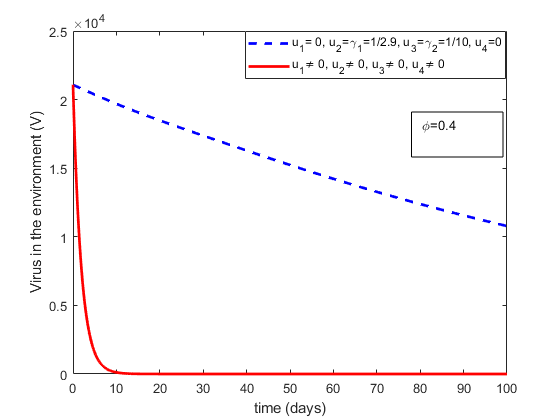}
  \includegraphics[width=0.3\textwidth]{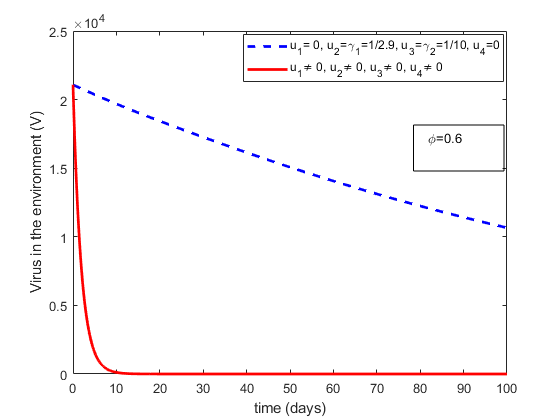}
  \includegraphics[width=0.3\textwidth]{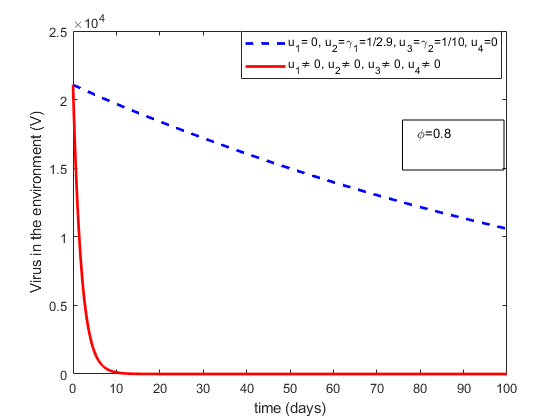}
  \caption{Controlling the viral spread in the environment with varying proportion of symptomatic individuals
  $\phi =0.4$, $\phi =0.6$ and $\phi=0.8$ where the red line represents the controlled environment and the blue
  line represents the uncontrolled environment.}
  \label{fig:virusinenvironment_C}
\end{figure}

%--------------------------------------------------
% \newpage
\subsection{Simulations results for all three optimal control strategies}
In this subsection, solution trajectories for the number of exposed, infectious with delay diagnosis and virus in the environment for all the three control strategies are numerically compared with that of the non-optimal control model. Our numerical results suggest that, if people can adhere to effective personal protection practices such as the use of hand sanitizers, washing of hands regularly and social distancing, there will less infections in the population. From our results, we can further argue that, effective spraying of the environment and early diagnosis of infected or infectious individuals and treatment can help reduce of the number of COVID-19 infections significantly.
 
\begin{figure}[!htbp]
  \centering
  \includegraphics[width=0.3\textwidth]{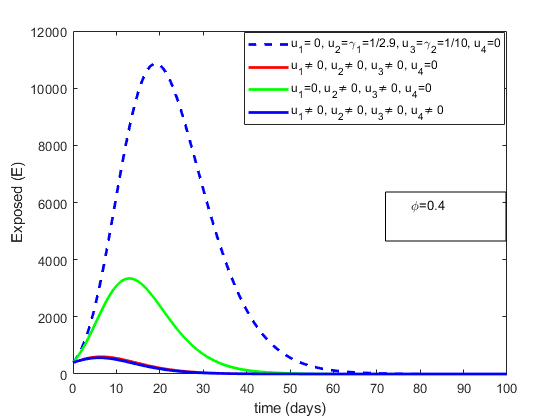}
  \includegraphics[width=0.3\textwidth]{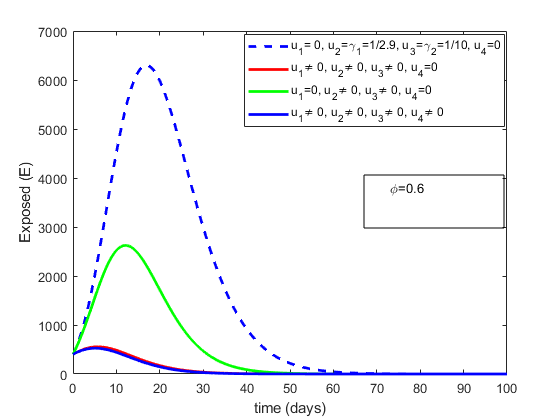}
  \includegraphics[width=0.3\textwidth]{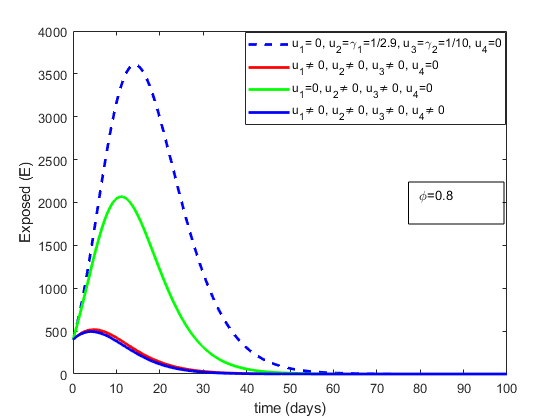}
  \caption{Solutions trajectories for Exposed individuals with $\phi =0.4$, $\phi =0.6$ and $\phi=0.8$ .}
  \label{fig:exposed_D}
\end{figure}

\begin{figure}[!htbp]
  \centering
  \includegraphics[width=0.3\textwidth]{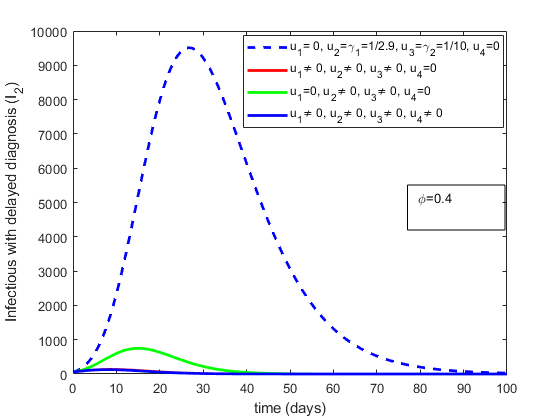}
  \includegraphics[width=0.3\textwidth]{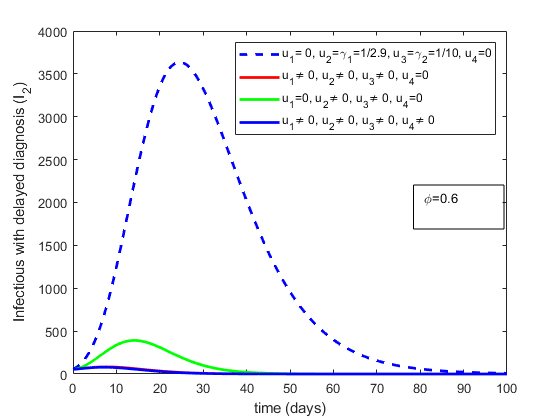}
  \includegraphics[width=0.3\textwidth]{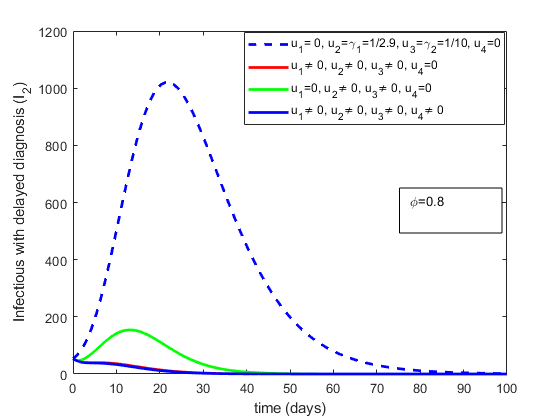}
  \caption{Solutions trajectories for Infectious individuals delayed diagnosis with $\phi =0.4$, $\phi =0.6$ and $\phi=0.8$ .}
  \label{fig:infectious_D}
\end{figure}

\begin{figure}[!htbp]
  \centering
  \includegraphics[width=0.3\textwidth]{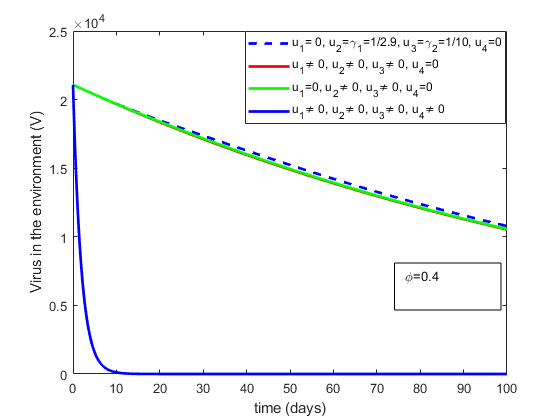}
  \includegraphics[width=0.3\textwidth]{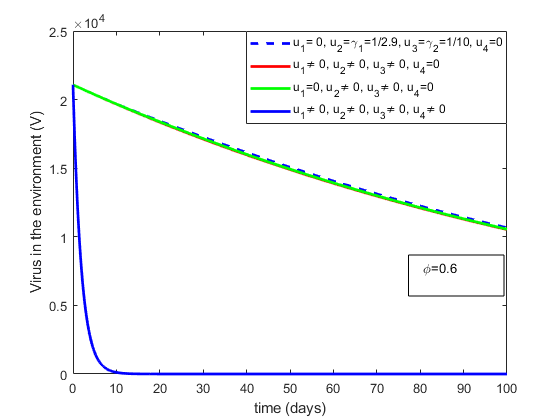}
  \includegraphics[width=0.3\textwidth]{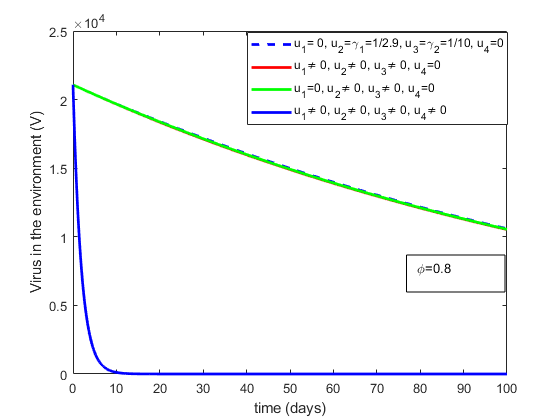}
  \caption{Solutions trajectories for virus in the environment with $\phi =0.4$, $\phi =0.6$ and $\phi=0.8$ .}
  \label{fig:virusinenvironment_D}
\end{figure}

%%%%%%%%%%%%%%%%%%%%%%%%%%%%%%%%%%%%%%%%%%%%%%

% \newpage
\section{Conclusion}
\label{sec:conclusion}
In this article, we presented control strategies to the transmission dynamics of COVID-19.
The control set including personal protection, treatment when individuals are early diagnosed, 
treatment when individuals are delay diagnosed,
effective spraying of the environment and cleaning possible infected surfaces can help reduce the quantity of the virus.
The numerical simulations reveals that optimal control strategies can yield significant reduction of the number of COVID-19
exposed and infectious or infected individuals in the population. By instituting control strategies, we realise that 
the number of days required for the virus to be eliminated from the system is significantly reduced as compared
to when there is no control strategy. The numerical illustrations also shows that by increasing $\phi$ i.e. improving 
the diagnostic resources, and increasing $\gamma_2$ i.e. improving the diagnostic efficiency, we can control 
significantly the number of new confirmed cases, new infections and thus can reduce the transmission risk. 
From all the three control strategies considered in this study, we realised that the third strategy which captures 
all the four time-dependent control functions yields better results.
%%%%%%%%%%%%%%%%%%%%%%%%%%%%%%%%%%%%%%%%%%%%%%%%%%%%%%%%%%%%%%%%%%%%%%%%%%%%%%%%%5
\bibliographystyle{plain} %Entries are ordered alphabetically;
\bibliography{OPtimalControlCOVID}
}
\end{document}